\documentclass[10pt,a4paper]{article}

\usepackage[latin3]{inputenc}
\usepackage{amsmath}
\usepackage{amsfonts}
\usepackage{amssymb}
\usepackage{amsthm}

\newtheorem{definition}{Definition}
\newtheorem{lemma}{Lemma}
\newtheorem{theorem}{Theorem}

\usepackage{fancyhdr}
\usepackage[authoryear]{natbib}
\bibpunct{(}{)}{,}{a}{}{;}

\DeclareMathOperator*{\argmax}{arg\,max}
\DeclareMathOperator*{\argmin}{arg\,min}

\newcommand{\bbc}{\,\widehat{\otimes}\,}
\newcommand{\ps}{\mathcal{P}}
\newcommand{\U}{\mathcal{U}}
\newcommand{\ignore}[1]{}

\begin{document}

\title{A generalisation of Nash's theorem with higher-order functionals}
\author{Julian Hedges \\
\emph{Queen Mary University of London, London E1 4NS, UK} \\
\texttt{jules.hedges@eecs.qmul.ac.uk}}
\date{}
\maketitle

\begin{abstract}
The recent theory of sequential games and selection functions by
Martin Escard\'o and Paulo Oliva is extended to games in which players
move simultaneously.
The Nash existence theorem for mixed-strategy equilibria of
finite games is generalised to games defined by selection functions.
A normal form construction is given which
generalises the game-theoretic normal form, and its soundness is proven.
Minimax strategies also generalise to the new class of
games and are computed by the Berardi-Bezem-Coquand functional,
studied in proof theory as an interpretation of the axiom of countable
choice.
\end{abstract}

\section{Introduction}

The notion of \emph{optimisation} is common to many areas of applied
mathematics, such as game theory and linear and nonlinear
programming. Typically we have a set $X$ of
\emph{choices} and a function $p$ mapping each $x \in X$ to a real
number $p(x)$, which we might call the \emph{value} or \emph{cost} of
$x$. From this we can define a natural notion of \emph{optimality}: a
point $y \in \mathbb{R}$ is optimal just if $y \geq p(x)$ for all $x
\in X$ and $y = p(x_0)$ for some $x_0 \in X$. We usually refer to $y$
by a notation such as
\[ y = \max_{x \in X} p(x) \]
The point $x_0$ is also interesting: it is a point at which $p$
\emph{attains} its optimal value, and we refer to it as
\[ x_0 = \argmax_{x \in X} p(x) \]
(Of course, while $y$ is guaranteed to be unique when it exists, $x_0$
is not necessarily unique; we only require that $\argmax$ chooses some
value for $x_0$.) These notations are connected by the equation $y = p(x_0)$, or
\[ \max_{x \in X} p(x) = p \left( \argmax_{x \in X} p(x)
\right) \]

Suppose we fix the set $X$ and assume that $\max_{x \in X}
p(x)$ exists for all functions $p : X \to \mathbb{R}$ (as when $X$ is
finite, for example). We can now define a function by
\[ \varphi (p) = \max_{x \in X} p(x) \]
$\varphi$ has range $\mathbb{R}$, and its domain is the function set $X
\to \mathbb{R}$, that is, the set of all functions with domain $X$ and
range $\mathbb{R}$. We therefore write
\[ \varphi : (X \to \mathbb{R}) \to \mathbb{R} \]
We call $\varphi$ a \emph{higher-order function}, that is, a function
whose domain is itself a set of functions. We can also define
\[ \varepsilon (p) = \argmax_{x \in X} p(x) \]
obtaining a higher-order function
\[ \varepsilon : (X \to \mathbb{R}) \to X \]
satisfying
\[ \varphi (p) = p (\varepsilon (p)) \hbox{ for all } p : X \to
\mathbb{R} \]

Using the concept of a higher-order function we can make a large
generalisation of the properties of $\max$ and
$\argmax$. For any sets $X$ and $R$, a function $\varphi : (X \to R)
\to R$ will be called a \emph{quantifier} and a function $\varepsilon
: (X \to R) \to X$ will be called a \emph{selection function}. We
say that $\varepsilon$ \emph{attains} $\varphi$ just if $\varphi (p) =
p (\varepsilon (p))$ for all $p : X \to R$. $\max$ and $\argmax$
become the prototypical examples of a quantifier and a selection function
attaining it. A very different example of a quantifier is a fixed
point operator $\mu : (X \to X) \to X$ which has the property that
$\mu(p)$ is always a fixed point of $p$, that is, $\mu(p) =
p(\mu(p))$. Thus a fixed point operator attains itself. Quantifiers
where $R$ is the set of truth-values appear naturally in logic. These concepts were
introduced and applied to the theory of sequential
games by Martin Escard\'o and Paulo Oliva in a series of papers
summarised in \citep{escardo11}. 

What is a game? Typically, some players take turns choosing between
sets of legal moves, which may be constrained by previous players'
moves. The sequence of moves made by the players is called a
\emph{play} of the game. Usually, the rules of the game guarantee that
every play terminates after a finite number of moves, and then
uniquely determine which player has won the play.

In the theory of games as introduced by \citep{neumann44}
the notion of a player winning a play is not
used. Rather, for each player the rules of the game define an
\emph{outcome function} mapping each play of the game to a real number
called the \emph{utility} of the play for that player. This
generalisation is important for applications of game theory to
economics, where utility often represents profit. In the game played by two
competing firms, for example, each firm is interested in maximising its own profit,
and does not care (in the short term, at least) how much profit its
competitor makes. Of course, a firm's profits will be affected by the
moves of its competitor, and vice versa. A central problem of
game theory is to determine which moves each player should
choose in order to maximise their utility. The theory of games as
surveyed for example in \citep{fudenberg91} will be referred to as
\emph{classical game theory}.

Suppose during the course of a play some player must choose between
some set $X$ of moves. Taking the usual assumption of common knowledge
of rationality (that is, the players play optimally, and they know
that each other will play optimally, and so on) the future of the play
after making each choice of $x \in X$ is sufficiently well determined
that a utility $p(x) \in \mathbb{R}$ can be assigned to each $x \in
X$. In classical game theory, a rational player will always choose
$\argmax_{x \in X} p(x)$. By replacing $\mathbb{R}$ with an
arbitrary set $R$ and argmax with an arbitrary selection function
$\varepsilon : (X \to R) \to X$, a rich theory of generalised
games results, with deep connections to proof theory and
theoretical computer science \citep{escardo10d, escardo11b}.

The games which have been described so far are the so-called
\emph{sequential} games. In the more usual language of classical game
theory this can be read as \emph{non-branching extensive form games of perfect
  information}. However there are games which cannot be
described as a sequence of moves. These are the so-called
\emph{simultaneous} games, or \emph{games of imperfect information}. A
well known example is
rock-paper-scissors; a more important example is the simultaneous
pricing of goods by supermarkets. von Neumann and
Morgenstern proved that every game can be described as a
simultaneous game, called its \emph{normal} or \emph{strategic form}. The central idea
of this proof is that players simultaneously choose \emph{contingent
  strategies}, higher-order functions which choose the next move given
the play up to that point, and so play the game on behalf of the
player. In this paper we consider a notion of simultaneous games that
encompasses Escard\'o and Oliva's generalised sequential games in a
similar way.

In section 3, generalised simultaneous games and their appropriate
notion of equilibrium are defined. In section 4 a class of games, the
so-called \emph{multilinear games}, is defined, and it is proven that
games of this kind always have an equilibrium (theorem
\ref{multilinear-existence-thm}). This is used in section 5 to prove
the key result of this paper (theorem \ref{generalised-nash-theorem}),
a natural generalisation of Nash's theorem for the existence of
mixed-strategy equilibria to games defined by arbitrary
quantifiers. In section 6 a mapping from sequential to simultaneous
games is defined analagous to the normal form construction in the
classical theory, and its soundness is proven (theorem
\ref{normal-form-soundness}). In section 7 we show an
interesting connection to proof theory, namely that the \emph{binary
  Barardi-Bezem-Coquand functional} computes minimax strategies of
games, a result that suggests a deeper connection between proof theory
and generalised games.

\section{Preliminaries}

If $X$ and $Y$ are sets then $X \to Y$ denotes the set of all
functions with domain $X$ and range $Y$ (this is often denoted $Y^X$,
a notation we avoid in order to avoid writing exponential towers for
higher-order functions). Cartesian products of sets are denoted $\prod$ and
bind tighter then $\to$, so for example $\prod_{i \in I} X_i \to R$
means $(\prod_{i \in I} X_i) \to R$. The $i$th coordinate projection
of a tuple $\pi \in \prod_{i \in I} X_i$ is denoted $\pi_i$.

The following piece of notation, for manipulating products, will be
helpful. Let $I$ be a set and let $X_i$ be a set for each $i \in I$. If $x \in X_i$
and $\pi \in
\prod_{j \in I} X_j$ then we define $\pi (i \mapsto x) \in \prod_{j \in I} X_j$ by
\[ (\pi(i \mapsto x))_j = \begin{cases}
x &\hbox{ if } i = j \\
\pi_j &\hbox{ otherwise}
\end{cases} \]

We make use of Church's $\lambda$-notation for describing functions
anonymously. The function which might otherwise be written as $x
\mapsto 1 + x$ will be denoted $\lambda x^\mathbb{N} . 1 + x$, where $\mathbb{N}$ is the
domain of the anonymous function. For example we have
$(\lambda x^\mathbb{N} . 1 + x)(42) = 43$. A variable bound by a
$\lambda$ need not appear under the scope of the $\lambda$, for
example $\lambda x^X . 42$ is the constant function with the property
that $(\lambda x^X . 42)(x') = 42$ for all $x' \in X$.

A \emph{quantifier} is a function $\varphi \in S_R (X)$ where $S_R (X) =
(X \to R) \to \mathcal P (R)$, a definition introduced in
\citep{escardo11}. The \emph{domain} of a quantifier is
\[ \mbox{dom} (\varphi) = \{ p \in X \to R \mid \varphi (p) \neq \varnothing
\} \]
A quantifier with $\mbox{dom} (\varphi) = X \to R$ will be called
\emph{total}.

A \emph{selection function} is a function $\varepsilon \in J_R (X)$
where $J_R (X) = (X \to R) \to X$. Selection functions were first introduced in
\citep{escardo10a}. The quantifier
$\varphi \in S_R (X)$ is \emph{attained} by the
selection function $\varepsilon \in J_R (X)$ just if
\[ p (\varepsilon (p)) \in \varphi (p) \]
for all $p \in \mbox{dom} (\varphi)$. This definition of attainment
differs from Escard\'o and Oliva's, who require the condition to
hold for all $p \in X \to R$. For a total
quantifier (which are considered in section 5, and to which the main
theorem applies) the two definitions coincide.

For example, if $R = \mathbb R$
and $X$ is compact then the extreme value theorem (plus the axiom of
choice) implies that the maximum quantifier
\[ \varphi (p) = \begin{cases}
\left\{ \max_{x \in X} p(x) \right\} &\mbox{ if } p \mbox{ is
  continuous} \\
\varnothing &\mbox{ otherwise}
\end{cases} \]
is attained. (Note that we need the axiom of choice to collect all the
values into a single function.) A quantifier such as this whose values
have cardinality at most 1 will be called \emph{single-valued}. Since
the definition of $\varphi$ is clumsy we can use a new notation for
single-valued quantifiers, such as
\[ \varphi (p) = \left. \max_{x \in X} p(x) \right|_{\mbox{$p$ is
    continuous}} \]

We assume some point-set topology as covered, for example, in \citep{kelley55}
and elementary properties of topological vector
spaces \citep{conway90}. All topological vector spaces are assumed to
be $T_1$ throughout (this is no loss of generality because quotienting
a topological vector space by the closure of $\{ 0 \}$ always yields a
Hausdorff space). For reference, a subset $S$ of a real vector
space is called \emph{convex} iff for all $x, y \in S$ and $t \in
[0,1]$ we have $tx + (1 - t)y \in S$.

In part 4 we work with the class of
\emph{locally convex spaces}. The definition of a locally convex space
is technical and not necessary for our purposes; beyond theorem
\ref{kakutani-theorem} and lemma \ref{locally-convex-products} we only
need to know that every locally convex
space is a topological vector space. Every normed vector
space is locally convex; examples of locally convex spaces which are
not normable include the spaces of smooth functions
$C^\infty(\mathbb{R})$ and $C^\infty([0,1])$ and the space of
real-valued sequences $\mathbb{R}^\omega$ with convergence defined
pointwise. Locally convex spaces are covered in detail in \citep{conway90}.

\bigskip

\textbf{A note on foundations.} It is possible to define generalised
sequential games over models other than classical set theory. Indeed,
as explained in \citep{escardo11} it is sometimes necessary to work in
nonstandard models, for example when considering unbounded sequential games
(which are not considered in this paper). $J_R$ is a (strong) monad
and can be defined over any cartesian closed category (moreover the
closely related $K_R (X) = (X \to R) \to R$, which contains the total
single-valued quantifiers, is already well-known from
programming language theory where it is called the \emph{continuation
  monad}). The definitions of
generalised simultaneous game and abstract Nash equilibrium could be
formalised in a more general setting, but the proofs in section 4 use
classical set theory in an
essential way, so we find it easier to avoid foundational issues
altogether and work entirely in classical set theory.

\section{Generalised simultaneous games}

In this section we define the objects studied in this paper, namely
\emph{generalised simultaneous games} and \emph{generalised Nash
equilibria}. The definition of a generalised simultaneous game comes from the classical
definition of a normal-form game, but with the maximising behaviour of
players replaced with a specified quantifier. For the general
definition we do not require the number of players to be finite. The
related notion of \emph{generalised sequential game} will be defined
in section 6.

\begin{definition}[Generalised simultanous game]
A \emph{generalised simultaneous game} (with multiple outcome spaces),
denoted simply \emph{game} when not ambiguous, is a tuple
\[ \mathcal G = (I, (X_i, R_i, q_i, \varphi_i)_{i \in I}) \]
where $I$ is a nonempty set of \emph{players}, and for each $i \in I$,
\begin{itemize}
\item $X_i$ is a nonempty set of \emph{moves} for player $i$;
\item $R_i$ is a set of \emph{outcomes} for player $i$;
\item $q_i \in S \to R_i$ is the \emph{outcome
    function} for player $i$, where $S = \prod_{j \in I} X_j$ is the
  \emph{strategy space} of $\mathcal G$;
\item $\varphi_i \in S_{R_i} (X_i)$ is the \emph{quantifier} for
  player $i$.
\end{itemize}

We say that $\mathcal G$ has a \emph{single outcome space} if the
$R_i$ are equal and the $q_i$ are equal. In this case $\mathcal G$ is
determined by a tuple
\[ \mathcal G = (I, (X_i)_{i \in I}, R, q, (\varphi_i)_{i \in I}) \]
where $q \in S \to R$.

An element $x
  \in X_i$ is called a \emph{strategy} for player $i$ for
  $\mathcal G$. A tuple $\pi \in S$ is called a \emph{strategy
    profile} for $\mathcal G$. Throughout this paper the variables
  $\pi$, $\sigma$ and $\tau$ will range over strategies of a game.
\end{definition}

In general we need games with multiple outcome spaces to study
simultaneous games, and in particular to recover the classical Nash
theorem. However normal forms of generalised sequential games will
always have single outcome space.

The appropriate notion of equilibrium of a generalised simultaneous
game is called a \emph{generalised Nash equilibrium}. Before making
this definition, we first define some notation used
throughout this paper. Firstly we define the family of
\emph{unilateral maps} $\U_q^i$, which are used as a shorthand
notation but, when considered as a higher-order functions, are also
natural and interesting in their own right.

\begin{definition}[Unilateral map]
Let $I$ be a set, and for each $i \in I$ let $X_i$ and $R_i$ be sets.
Let $q = (q_i)_{i \in I}$ be a family of maps such that each
\[ q_i \in \prod_{j \in I} X_j \to R_i \]
We define the $i$th \emph{unilateral map}
\[ \U_q^i \in \prod_{j \in I} X_j \to (X_i \to R_i) \]
by
\[ \U_q^i (\pi)(x) = q_i(\pi(i \mapsto x)) \]
\end{definition}

Thus, the $i$th unilateral map computes the outcomes of unilateral
changes of strategy by the $i$th player in a game. Secondly, we
associate to every quantifier a set called its \emph{diagonal}.

\begin{definition}[Diagonal of a quantifier]
Let $\varphi \in S_R (X)$ be a quantifier. The \emph{diagonal}
of $\varphi$ is
\[ \Delta(\varphi) = \{ (p, x) \in (X \to R) \times
X \mid p(x) \in \varphi(p)
\} \]
\end{definition}

Now the equilibria of a generalised simultaneous game can be defined
in a very compact and (as will be seen) useful way.

\begin{definition}[Generalised Nash equilibrium]
Let $\mathcal G$ be a game with strategy space $S$. We define the
\emph{best response correspondence} $B \in S \to \ps (S)$ of $\mathcal
G$ by
\[ B(\pi) = \bigcap_{i \in I} B_i (\pi) \]
where the $B_i \in S \to \ps (S)$ are defined by
\[ B_i (\pi) = \{ \sigma \in S \mid (\U_q^i (\pi), \sigma_i) \in
\Delta(\varphi_i) \} \]

A \emph{generalised Nash equilibrium} of $\mathcal G$ is a fixed point
of $B$, that is, a strategy profile $\pi$ such that $\pi \in B(\pi)$.
\end{definition}

Unpacking this definition, we see that $\pi$ is a generalised Nash
equilibrium of $\mathcal G$ iff for each $i \in I$ we have
\[ q_i (\pi) \in \varphi_i (\lambda x^{X_i} . q_i (\pi (i \mapsto
x))) \]
When $X_i$ is compact, $q_i$ is continuous and $\varphi_i$ is the
quantifier
\[ \varphi_i (p) = \left. \max_{x \in X_i} p(x) \right|_{\mbox{$p$ is
    continuous}} \]
this reduces to
\[ q_i (\pi) = \max_{x \in X_i} q_i (\pi (i \mapsto x)) \]
which is the usual definition of a Nash equilibrium.

\section{Multilinear games}

Now we define a large family of games, called the \emph{multilinear
  games}, that are guaranteed to have a
generalised Nash equilibrium. The structure of the argument is the same
as that in \citep{nash50b}, but given in more generality to deal with
more general quantifiers. This section can be seen as a series of
lemmas that are eventually used to prove theorem
\ref{generalised-nash-theorem} (the generalisation of Nash's theorem)
in the next section.

\begin{definition}[Closed graph property]
Let $X$ and $Y$ be topological spaces and $F \in X \to \ps (Y)$. We say that
$F$ has the \emph{closed graph property} iff
\[ \Gamma(F) = \{ (x, y) \in X \times Y \mid y \in F(x) \} \]
is closed with respect to the product topology.
\end{definition}

The closed graph property is a form of continuity for functions whose
range is a set of subsets of a topological space.

In order to guarantee that a generalised simultaneous game will have
an equilibrium we need to impose closed graph properties on the
quantifiers. However the domain of a quantifier is a function set, which in general
has no unique natural topology. The least we need is that the
unilateral maps are continuous, and so for this reason we define the
\emph{unilateral topology}.

\begin{definition}[Unilateral topology]
For each $i \in I$ let $X_i$ and $R_i$ be topological spaces with $q_i \in
X_i \to R_i$ continuous. The \emph{unilateral topology} on $X_i \to
R_i$ is the final topology with respect to the singleton family $\{
\U_q^i \}$, that is, it is the largest topology with respect to which
$\U_q^i$ is continuous. A function which is continuous with respect to
the unilateral topology
will be called \emph{unilaterally continuous}, and a function which
has closed graph with respect to the unilateral topology has
\emph{unilaterally closed graph}.
\end{definition}

Another possible topology on $X_i \to R_i$ which will be useful is the
topology of pointwise convergence. Most of this paper
could be formulated using only pointwise convergence, except for an
interesting example at the end of this section which needs a finer
topology, namely uniform convergence.

\begin{lemma}
The unilateral topology is finer than the topology of pointwise convergence.
\end{lemma}

\begin{proof}
It must be proven that $\U_q^i$ is continuous with respect to the
topology of pointwise convergence. Let $\pi_j \longrightarrow \pi$ be a convergent
sequence in $\prod_{j \in I} X_j$, and let $x \in X_i$. We have
\[ \pi_j (i \mapsto x) \longrightarrow \pi(i \mapsto x) \]
in the product topology, so
\[ \U_q^i (\pi_j)(x) = q_i(\pi_j (i \mapsto x)) \longrightarrow q_i(\pi (i \mapsto x))
= \U_q^i (\pi)(x) \]
because $q_i$ is continuous. Therefore $\U_q^i (\pi_j) \longrightarrow \U_q^i (\pi)$
pointwise, as required.
\end{proof}

Now we can give the definition of a multilinear game. This definition
essentially contains the least assumptions needed for Nash's proof.

\begin{definition}[Multilinear game]
A game $\mathcal G = (I, (X_i, R_i, q_i, \varphi_i)_{i \in I})$ is
called \emph{multilinear} iff
\begin{itemize}
\item Each $X_i$ is a compact and convex subset of a given locally
  compact space $V_i$ over $\mathbb R$;
\item Each $R_i$ is a topological vector space over $\mathbb R$;
\item Each $q_i$ extends to a continuous multilinear map
\[ q_i \in \prod_{j \in I} V_j \to R_i \]
(that is, $q_i$ is linear with respsect to each $V_j$ separately);
\item Each $\varphi_i$ has unilaterally closed graph, $\varphi_i
  (p)$ is closed and convex for all $p \in X_i \to R_i$, and
$\mbox{\emph{dom}} (\varphi_i) \supseteq \mbox{\emph{im}} (\U_q^i)$.
\end{itemize}
\end{definition}

(Note that because $q_i$ is continuous and multilinear, to satisfy the
last condition it suffices that $\varphi_i (p) \neq \varnothing$
whenever $p$ is continuous and linear. Note also that if $\varphi_i$
is single-valued then each $\varphi_i (p)$ is automatically closed and
convex.)

The idea of the existence proof is to reduce to the following fixed
point theorem.

\begin{theorem}[Kakutani-Fan-Glicksberg fixed point theorem
  \citep{fan52, glicksberg52}]\label{kakutani-theorem}
Let $S$ be a nonempty, compact and convex subset of a locally convex
space over $\mathbb{R}$. Let $B \in S \to \mathcal{P} (S)$ have closed
graph and let $B (\pi)$ be nonempty, closed and convex for all $\pi
\in S$. Then $B$ has a fixed point.
\end{theorem}

We will need to use the fact that locally convex spaces are closed
under arbitrary products.

\begin{lemma}\label{locally-convex-products}
Let $\{ V_i \}_{i \in I}$ be a family of locally convex spaces over a field $K$. Then
$\prod_{i \in I} X_i$ has the strucutre of a locally convex space over $K$
whose topology is the product topology.
\end{lemma}

Much of the usefulness of multilinear games comes down to the fact
that their unilateral maps are well-behaved.

\begin{lemma}
Let $\mathcal G$ be a multilinear game with strategy space $S$. Then
each $\U_q^i$ is a continuous function
\[ \U_q^i \in S \times X_i \to R_i \]
(under the Curry bijection $A \to (B \to C) \cong A \times B \to C$)
which is linear in its second argument.
\end{lemma}

\begin{proof}
By the continuity and multilinearity of the $q_i$.
\end{proof}

Lemmas \ref{first-existence-lemma}-\ref{last-existence-lemma} form the
core of the proof, establishing the hypotheses of the
Kakutani-Glicksberg-Fan theorem.

\begin{lemma}\label{first-existence-lemma}
Let $\mathcal{G}$ be a multilinear game. Then the strategy
space of $\mathcal{G}$ is a nonempty, compact and convex subset of a
locally convex space.
\end{lemma}

\begin{proof}
The strategy space is
\[ S = \prod_{i \in I} X_i \subseteq \prod_{i \in I} V_i \]
where the larger space is locally convex by lemma
\ref{locally-convex-products}.

$S$ is nonempty by the axiom of choice and compact by Tychonoff's
theorem. Convexivity is also inherited by the product, since for each
$i \in I$ we have
\[ (tx + (1 - t)y)_i = t x_i + (1 - t) y_i \in X_i \qedhere \]
\end{proof}

\begin{lemma}
Let $\mathcal G$ be a multilinear game with best response
correspondence $B$ such that each quantifier $\varphi_i$ is attained
by a selection function $\varepsilon_i$. Then $B(\pi)$ is nonempty for
all $\pi$.
\end{lemma}

\begin{proof}
Let $S$ be the strategy space of $\mathcal{G}$. Given $\pi \in S$ we
define $\sigma \in S$ to have $i$th component
\[ \sigma_i = \varepsilon_i (\U_q^i (\pi)) \]
Since $\varepsilon_i$ attains $\varphi_i$ and $\U_q^i (\pi) \in
\mbox{dom} (\varphi_i)$ we have
\[ \U_q^i (\pi) (\varepsilon_i (\U_q^i (\pi))) \in \varphi_i (\U_q^i
(\pi)) \]
Therefore
\[ (\U_q^i (\pi), \sigma_i) = (\U_q^i (\pi), \varepsilon_i (\U_q^i
(\pi))) \in \Delta(\varphi_i) \]
so $\sigma \in B(\pi)$, as required.
\end{proof}

\begin{lemma}
Let $\mathcal G$ be a multilinear game with best response
correspondence $B$. Then $B(\pi)$ is closed for all $\pi$.
\end{lemma}

\begin{proof}
It suffices to prove that each factor
\[ B_i (\pi) = \{ \sigma \in S \mid (\U_q^i (\pi), \sigma_i) \in
\Delta(\varphi_i) \} \]
is closed. Let $\sigma_j \longrightarrow \sigma$ be a convergent sequence in $B_i
(\pi)$. For each $i$ we have $\sigma_{j, i} \longrightarrow \sigma_i$,
so
\[ \U_q^i (\pi) (\sigma_{j, i}) \longrightarrow \U_q^i (\pi)
(\sigma_i) \]
by the continuity of $\U_q^i$. We also have that each
\[ \U_q^i (\pi)(\sigma_{j, i}) \in \varphi_i (\U_q^i (\pi)) \]
and the right hand side is closed by definition, therefore
\[ \U_q^i (\pi)(\sigma_i) \in \varphi_i (\U_q^i (\pi)) \]
that is,
\[ (\U_q^i (\pi), \sigma_i) \in \Delta(\varphi_i) \qedhere \]
\end{proof}

\begin{lemma}
Let $\mathcal G$ be a multilinear game with best response
correspondence $B$. Then $B(\pi)$ is convex for all $\pi$.
\end{lemma}

\begin{proof}
Suppose $\sigma, \tau \in B(\pi)$ and $t \in [0,1]$. Let $i \in I$. By
definition we have
\[ \U_q^i (\pi)(\sigma_i), \U_q^i (\pi)(\tau_i) \in \varphi_i (\U_q^i
(\pi)) \]
Since the linearity of $\U_q^i$ we have
\[ \U_q^i (\pi) (t \sigma_i + (1 - t) \tau_i) = t\,\U_q^i
(\pi)(\sigma_i) + (1 - t) \U_q^i (\pi)(\tau_i) \]
Since the $\varphi_i (p)$ are convex, we have
\[ \U_q^i (\pi) (t \sigma_i + (1 - t) \tau_i) \in \varphi_i (\U_q^i
(\pi)) \]
that is,
\[ (\U_q^i (\pi), t \sigma_i + (1 - t) \tau_i) \in \Delta
(\varphi_i) \]
Therefore
\[ t \sigma + (1 - t) \tau \in B(\pi) \qedhere \]
\end{proof}

\begin{lemma}\label{last-existence-lemma}
Let $\mathcal G$ be a multilinear game with best response
correspondence $B$. Then $B$ has closed graph.
\end{lemma}

\begin{proof}
Note that
\[ \Gamma(B) = \bigcap_{i \in I} \{ (\sigma, \pi) \in S^2 \mid (\U_q^i
(\pi), \U_q^i (\pi)(\sigma_i)) \in \Gamma (\varphi_i) \} \]
and so it suffices to prove these factors closed. Let $(\sigma_j,
\pi_j) \longrightarrow (\sigma, \pi)$ be a convergent sequence in the $i$th
factor. By the continuity of $\U_q^i$,
\[ \U_q^i (\pi_j) (\sigma_{j, i}) \longrightarrow \U_q^i (\pi)
(\sigma_i) \]

Since $\U_q^i$ is also unilaterally continuous as a map $S \to (X_i
\to R_i)$, we have $\U_q^i (\pi_j) \longrightarrow
\U_q^i (\pi)$ unilaterally. Therefore we have a convergent sequence 
\[ (\U_q^i (\pi_j), \U_q^i (\pi_j)(\sigma_{j, i})) \longrightarrow (\U_q^i (\pi),
\U_q^i (\pi)(\sigma_i)) \]
in the graph $\Gamma(\varphi_i)$, which is closed by definition.
\end{proof}

\begin{theorem}[Existence theorem for multilinear
  games]\label{multilinear-existence-thm}
Let $\mathcal{G}$ be a multilinear game such that each quantifier is
attained by a selection function. Then $\mathcal{G}$ has a generalised
Nash equilibrium.
\end{theorem}

\begin{proof}
Let $B$ be the best response correspondence of $\mathcal{G}$. By
lemmas \ref{first-existence-lemma}-\ref{last-existence-lemma} and the
Kakutani-Fan-Glicksberg fixed point theorem, $B$ a fixed point.
\end{proof}

Examples of multilinear games as mixed extensions of finite games are
given in the next section. Another interesting example is given by
integration. Let $X_i = [0, 1]$, $V_i = \mathbb R$ and $R_i = \mathbb R$, and let
$L(X_i)$ be the set of all Lebesgue-integrable functions $p \in [0, 1] \to
\mathbb R$ with
\[ \left| \int_0^1 p(x)\,dx \right| < \infty \]
Define a single-valued quantifier $\varphi_i \in S_{\mathbb R} X_i$ by
\[ \varphi_i (p) = \left. \int_0^1 p(x)\,dx \right|_{p \in L(X_i)} \]
Using the mean value theorem (and the axiom of choice) we can prove
the existence of a selection function attaining $\varphi_i$: for all
$p \in L(X_i)$ there exists $\varepsilon_i (p) \in X_i$ such that
\[ p(\varepsilon_i (p)) = \int_0^1 p(x)\,dx \]
This highly nonconstructive selection function was briefly introduced
as an example in \citep{escardo10a}.

We let $I$ be finite and for simplicity let the other $X_j$ be normed, so the strategy
space is normed and we can work with the $\varepsilon-\delta$
definitions of uniform convergence and continuity.

\begin{lemma}
If $q_i$ is uniformly continuous and $\pi_j \longrightarrow \pi$ then $\U_q^i
(\pi_j) \longrightarrow \U_q^i(\pi)$ uniformly.
\end{lemma}

\begin{proof}
We have that $q_i$ is uniformly continuous, that is,
\begin{equation}
\forall \varepsilon > 0\ \exists \delta > 0\ \forall \pi, \sigma \in
\prod_{j \in I} X_j .\ |\pi - \sigma| < \delta \implies |q_i(\pi) -
q_i(\sigma)| < \varepsilon
\end{equation}
We also have $\pi_j \longrightarrow \pi$, that is,
\begin{equation}
\forall \varepsilon > 0\ \exists N\ \forall j \geq N .\ |\pi_j -
\pi| < \varepsilon
\end{equation}
We want to prove that $\U_q^i (\pi_j) \longrightarrow \U_q^i (\pi)$ uniformly,
that is,
\[ \forall \varepsilon > 0\ \exists N\ \forall j \geq N\ \forall x \in X_i .\
|\U_q^i(\pi_j)(x) - \U_q^i(\pi)(x)| < \varepsilon \]

Let $\varepsilon > 0$. By (1), we have $\delta > 0$ with the given
property. We take $\varepsilon$ in (2) to be this $\delta$, obtaining
$N$. Let $j \geq N$, therefore
\[ |\pi_j - \pi| < \delta \]
by (2). Let $x \in X$. The crucial observation is that $\pi_j(i
\mapsto x)$ behaves like $\pi_j$ but is constant in its $i$th
coordinate. That is, we have
\[ |\pi_j (i \mapsto x) - \pi(i \mapsto x)| \leq |\pi_j - \pi| <
\delta \]

Now we take $\pi$, $\sigma$ in (1) to be $\pi_j (i \mapsto x)$ and
$\pi (i \mapsto x)$. We have already proved the antecedent in (1),
therefore
\[ |\U_q^i (\pi_j)(x) - \U_q^i (\pi)(x)| = |q_i(\pi_j (i \mapsto x)) -
q_i(\pi(i \mapsto x))| < \varepsilon \]
as required.
\end{proof}

We have proven that the unilateral topology is finer than the topology
of uniform convergence.

\begin{lemma}
$\varphi_i$ is unilaterally continuous.
\end{lemma}

\begin{proof}
Suppose we have $p_j \longrightarrow p$ uniformly in $L(X_i)$.
Since the convergence of the integrands is uniform, we can apply the
uniform convergence theorem to get
\[ \varphi_i (p_j) = \int_0^1 p_j (x)\,dx \longrightarrow \int_0^1
p(x)\,dx = \varphi_i (p) \]
Since the unilateral topology is finer than the topology of uniform
convergence, we are done.
\end{proof}

In the 1-player game defined by the integration quantifier with
outcome function $q$, the unique value of $q(x)$ when $x$ is an
equilibrium strategy, which can be called the \emph{expected outcome}
of the game, is simply
\[ \int_0^1 q(x)\,dx \]
In the 2-player game with both quantifiers integrals, a generalised
Nash equilibrium $(a, b)$ satisfies
\[ a = \int_0^1 q_X (x, b)\,dx \qquad b = \int_0^1 q_Y (a, y)\,dy \]
Since $\int_0^1 p(x)\,dx$ is the \emph{average} value of $p$, this is
a game where players are trying to gain the average outcome rather
than the maximum.
The existence of a Nash equilibrium in this case can be more directly
proven by applying the Brouwer fixed point theorem to the mapping
\[ [0, 1]^2 \to [0, 1]^2,\ (a, b) \mapsto \left( \int_0^1 q_X (x,
  b)\,dx, \int_0^1 q_Y (a, y)\,dy \right) \]

\section{Finite games}

In this section we apply the existence theorem for multilinear games
to prove a suitable generalisation of Nash's theorem. The classical
version of Nash's theorem guarantees that every finite game (that is,
a classical game in which each player has finitely many strategies)
has a \emph{mixed strategy Nash equilibrium}.

The notion of mixed strategies means that we consider probably
distributions over ordinary strategies (referred to as \emph{pure
  strategies} for clarity). The outcome functions also need to be
replaced by \emph{expected outcome} functions. However the discussion
of probability distributions can be avoided by treating them as
geometric objects, namely \emph{simplices}. This approach also makes
it clearer how quantifiers and selection functions must be modified
when passing to mixed strategies. A probabilistic interpretation of
the resulting theorem is possible, but is avoided in this paper.

\begin{definition}[Finite game]
A game $\mathcal G = (I, (X_i, R_i, q_i, \varphi_i)_{i \in I})$ is
called \emph{finite} iff
\begin{itemize}
\item $I$ is finite;
\item Each $X_i$ is finite;
\item Each $R_i$ is a topological vector space over $\mathbb R$;
\item Each $\varphi_i$ is total, has closed graph with respect to the topology
  of pointwise convergence, and $\varphi_i (p)$ is closed and convex
  for all $p \in X_i \to R_i$.
\end{itemize}
\end{definition}

Note that restricting to pointwise convergence is no loss of
generality here because the $X_i$ are finite.

The set of probability distributions over a finite set can be seen as
a geometric object called a \emph{standard simplex}. In 2 and 3
dimensions these can be easily visualised as a line segment and an
equilateral triangle; the next simplex $\Delta_4$ is a tetrahedron
seen as a subset of $\mathbb R^4$.

\begin{definition}[Standard simplex]
The $n$th \emph{standard simplex} is the set
\[ \Delta_n = \{ (x_1, \ldots, x_n) \in \mathbb{R}^n \mid \sum_{i
  =1}^n x_i = 1 \hbox{ and each } x_i \geq 0 \} \]
\end{definition}

\begin{definition}[Mixed extension]
Let $\mathcal G = (I, (X_i, R_i, q_i, \varphi_i)_{i \in I})$ be a
finite game with strategy space $S$. We define a game
\[ \mathcal G^* = (I, (X_i^*, R_i, q_i^*, \varphi_i^*)_{i \in I}) \]
called the \emph{mixed extension} of $\mathcal G$ as follows: player
$i$ has move set
\[ X_i^* = \Delta_{|X_i|} \]
outcome function
\[ q_i^* (\pi) = \sum_{\sigma \in S} \left( \prod_{i \in I} \pi_{i,
    \sigma_i} \right) (q_i (\sigma)) \]
and quantifier
\[ \varphi_i^* (p) = \varphi_i (p \circ \delta_i) \]
where $\delta_i$ is the canonical
injection $X_i \hookrightarrow X_i^*$ mapping each $j$ to the
vertex of the simplex at which the $j$th coordinate is $1$.
\end{definition}

Note that the finiteness of $I$ is used only in the well-definition of
the $q_i^*$: for the strategy space $S$ to be finite
it is necessary that $I$ be finite, except in trivial cases when all
but finitely many $X_i$ are singletons.

\begin{definition}[Mixed strategy abstract Nash equilibrium]
Let $\mathcal G$ be a finite game. A strategy profile for $\mathcal
G^*$ will be called a \emph{mixed strategy profile} for $\mathcal
G$. An abstract Nash equilibrium of $\mathcal G^*$ will be called a
\emph{mixed strategy abstract Nash equilibrium} of $\mathcal G$.
\end{definition}

The most important property of mixed extensions is that they are
always multilinear. This will used to prove the generalised Nash
theorem.

\begin{lemma}\label{mixed-extensions-multilinear}
Let $\mathcal G$ be a finite game. Then $\mathcal G^*$ is a
multilinear game.
\end{lemma}

\begin{proof}
Each $\Delta_n$ for $n > 0$ is a nonempty, compact and convex subset
of the locally convex space $\mathbb R^n$. Continuity of the $q_i^*$
is clear. $q_i^*$ is multilinear because
\begin{align*}
&\U_{q^*}^i (\pi)(cx + y) \\
=\ &q^*(\pi(i \mapsto cx + y)) \\
=\ &\sum_{\sigma \in S} \left( \left( \prod_{j \neq i} \pi_{j,
      \sigma_j} \right) \cdot (cx_{\sigma_i} + y_{\sigma_i}) \right)
(q (\sigma)) \\
=\ &c \sum_{\sigma \in S} \left( \left( \prod_{j \neq i} \pi_{j,
      \sigma_j} \right) \cdot x_{\sigma_i} \right) (q (\sigma)) +
  \sum_{\sigma \in S} \left( \left( \prod_{j \neq i} \pi_{j, \sigma_j}
      \right) \cdot y_{\sigma_i} \right) (q (\sigma)) \\
=\ &cq^*(\pi(i \mapsto x)) + q^*(\pi(i \mapsto y)) \\
=\ &c\,\U_{q^*}^i (\pi)(x) + \U_{q^*}^i (\pi)(y)
\end{align*}

The $\varphi_i^* (p)$ are of the form $\varphi_i (p')$, and so are
closed and convex. We note that $\varphi_i^*$ is total because
$\varphi_i$ is. The graph of $\varphi_i^*$ is
\[ \Gamma (\varphi_i^*) = \{ (p, y) \in (X_i^* \to R_i) \times R_i \mid
(p \circ \delta_i, y) \in \Gamma (\varphi_i) \} \]
Suppose we have a convergent sequence $(p_j, y_j) \longrightarrow (p, y)$ in
$\Gamma (\varphi_i^*)$. Let $x \in X$, then
\[ (p_j \circ \delta_i)(x) = p_j (\delta_i (x)) \longrightarrow p (\delta_i (x)) = (p
\circ \delta_i)(x) \]
Therefore $p_j \circ \delta_i \longrightarrow p \circ \delta_i$ pointwise, so
\[ (p_j \circ \delta_i, y_j) \longrightarrow (p \circ \delta_i, y) \]
with respect to the topology of pointwise convergence. Since the
unilateral topology is finer than the topology of pointwise convergence, we
are done.
\end{proof}

The final result we need is the ability to lift selection functions to
mixed extensions.

\begin{lemma}\label{mixed-selection-functions}
Let $X$ be a nonempty finite set and let $\varphi \in S_{\mathbb R} X$
be a total quantifier attained by the selection function $\varepsilon
\in J_{\mathbb R} X$. Then there exists a selection function $\varepsilon^*$ such
that $\varphi^*$ is attained by $\varepsilon^*$.
\end{lemma}

\begin{proof}
We define $\varepsilon^* \in J_{\mathbb R} \Delta_{|X|}$ by the equation
\[ \varepsilon^* (p) = \delta (\varepsilon (p \circ \delta)) \]
where $\delta : X \hookrightarrow \Delta_{|X|}$. Then
\[ p (\varepsilon^* (p)) = p (\iota (\varepsilon (p \circ \iota)) = (p
\circ \iota) (\varepsilon (p \circ \iota)) \in \varphi (p \circ \iota)
= \varphi^* (p) \]
for all $p \in \Delta_{|X|} \to \mathbb R$.
\end{proof}

\begin{theorem}[Existence theorem for finite
  games]\label{generalised-nash-theorem}
Let $\mathcal{G}$ be a finite game such that each $\varphi_i$ is
attained by a selection function. Then $\mathcal{G}$ has a mixed
strategy abstract Nash
equilibrium.
\end{theorem}

\begin{proof}
$\mathcal{G}^*$ is a multilinear game which is attained by selection
functions by lemmas \ref{mixed-extensions-multilinear} and
\ref{mixed-selection-functions}. Therefore $\mathcal{G}^*$ has an
abstract Nash equilibrium by theorem \ref{multilinear-existence-thm}.
\end{proof}

In order to recover the classical Nash theorem we simply consider
finite games whose outcome spaces are $\mathbb R$ and define $q_i
(\pi) \in \mathbb R$ to be the utility of $\pi$ for player $i$, taking
all quantifiers to be $\max$. We
could instead define a finite game with single outcome space $\mathbb
R^n$ and let $(q(\pi))_i$ be the utility of $\pi$ for player $i$, and
consider selection functions $\varepsilon_i \in J_{\mathbb R^n} X_i$
maximising the $i$th coordinate:
\[ \varepsilon_i (p) = \argmax_{x \in X_i} (p(x))_i \]
However the quantifiers attained by these quantifiers are continuous only if $n = 1$. This
game has the same equilibria as the equivalent game with multiple
outcome spaces, but the Nash theorem cannot be proven in this way. It
is for this reason that we need to consider games with multiple
outcome spaces, in contrast to generalised sequential games (which do
not require continuity).

For a different example of a quantifier in a finite game, let $R_i$ be
normed and fix $\varepsilon > 0$ and $x_0 \in X_i$. Define
\[ \varphi_i (p) = B_\varepsilon (p(x_0)) \]
that is, the closed $\varepsilon$-ball around $p(x_0)$. This quantifier
is attained by the constant selection function $\varepsilon(p) =
x_0$. For a sequential game this would force the game to be trivial, but
this is not the case here: for example, if $\varphi_X$ is the quantifier defined
here and $\varphi_Y$ is the maximum quantifier with $R_Y = \mathbb R$
then a Nash equilibrium is a point $(a, b)$ such that
\[ \left| q_X (a, b) - q_X (x_0, b) \right| < \varepsilon \qquad
q_Y (a, b) = \max_{y \in Y} q_Y (a, y) \]

\section{The normal form of a sequential game}

In classical game theory every game can be put into the form of a
simultaneous game called its \emph{normal form}. The major motivation
for definining generalised simultaneous games was to generalise this
operation to give a notion of normal form for generalised sequential
games. This construction is given in this section and a form of
soundness of proven, namely that the solution concept for a
generalised sequential game, the so-called \emph{optimal strategies},
are mapped to generalised Nash equilibria.

\begin{definition}[Generalised sequential game]
A \emph{generalised sequential game} is determined by a set $R$ of \emph{outcomes}, a set
$X_i$ of \emph{moves} and a quantifier $\varphi_i \in S_R X_i$ for each $1
\leq i \leq n$, and an \emph{outcome function} $q \in \prod_{i = 1}^n X_i \to
R$. A \emph{strategy} in a sequential game is a tuple
\[ \pi \in \prod_{i = 1}^n \left( \prod_{j = 1}^{i - 1} X_j \to X_i
\right) \]
The strategy $\pi$ is called \emph{optimal} iff for all $\vec{a} =
(a_1, \ldots, a_{i - 1}) \in \prod_{j = 1}^{i - 1} X_j$ (where $i >
0$) we have
\[ q(\vec{a}, b^{\vec{a}}_i, \ldots, b^{\vec{a}}_n) \in
\varphi_i (\lambda x^{X_i} . q(\vec{a}, x, b^{\vec{a}, x}_{k + i},
\ldots, b^{\vec{a}, x}_n)) \]
where
\[ b^{\vec{a}}_j = \pi_j (\vec{a}, b^{\vec{a}}_i, \ldots,
b^{\vec{a}}_{j - 1}) \]
Given a strategy $\pi$ in a game, its \emph{strategic play} is
$\pi^\dag \in \prod_{i = 1}^n X_i$ given by
\begin{align*}
\pi^\dag_1 &= \pi_1\ (\hbox{modulo the isomorphism } \prod_{j = 1}^0
X_j \to X_1 = \{ 0 \} \to X_1 \cong X_1) \\
\pi^\dag_{i + 1} &= \pi_{i + 1} (\pi^\dag_1, \ldots, \pi^\dag_i)
\end{align*}
The strategic play of an optimal strategy is called an \emph{optimal
  play}.
\end{definition}

To be precise, this notion of sequential game is called a \emph{finite
  game with multiple optimal outcomes} in \citep{escardo11}. Infinite
games are avoided in this paper for simplicity.

Generalised sequential games were introduced in order to study a
particular higher-order function called the \emph{product of selection
  functions}. This is the function
\[ \otimes \in J_R X \times J_R Y \to J_R (X \times Y) \]
given by
\[ (\varepsilon \otimes \delta)(q) = (a, b_a) \]
where
\[ a = \varepsilon (\lambda x^X . q(x, b_x)) \qquad b_x = \delta
(\lambda y^Y . q(x, y)) \]
The product of selection functions has many interesting and
unintuitive properties, especially when infinitely iterated: for
example, it computes witnesses for the axiom of countable choice, and
computes exhaustive searches of certain infinite types in finite time
\citep{escardo10d}, both of which popular belief would have is
impossible. Every use of the product of selection
functions can be seen as the compution of an optimal play for a
suitable generalised sequential game.

\begin{theorem}\label{optimal-strategies}
Let $\mathcal G$ be a generalised sequential game whose quantifiers
are total and attained by selection functions $\varepsilon_i \in J_R
X_i$. Then
\[ \left( \bigotimes_{i = 1}^n \varepsilon_i \right) (q) \]
is an optimal play for $\mathcal G$.
\end{theorem}

Now we give the normal form construction and prove that it maps
optimal strategies to generalised Nash equilibria.

\begin{definition}[Normal form]
Let $\mathcal{G}$ be a generalised sequential game. We define a simultaneous game
with single outcome space
\[ \mathcal G^\dag = (I^\dag, (X_i^\dag)_{i \in I}, R, q^\dag,
(\varphi_i^\dag)_{i \in I}) \]
called the \emph{normal form} of $\mathcal{G}$ as follows:
\begin{itemize}
\item $I^\dag = \{ 1, \ldots, n \}$;
\item Each $X_i^\dag = P_i \to X_i$ where $P_i = \prod_{j = 1}^{i - 1} X_j$;
\item $q^\dag (\pi) = q (\pi^\dag)$;
\item Each $\varphi_i^\dag (p) = \varphi_i (\lambda x^{X_i} . p
  (\lambda \sigma^{P_i} . x))$.
\end{itemize}
\end{definition}

\begin{theorem}\label{normal-form-soundness}
Let $\mathcal{G}$ be a sequential game and let $\pi$
be an optimal strategy for $\mathcal{G}$. Then $\pi$ is a generalised
Nash equilibrium of $\mathcal{G}^\dag$.
\end{theorem}

\begin{proof}[Proof]
Let $1 \leq i \leq n$. It must be proven that
\[ (\U_{q^\dag}^i (\pi), \pi_i) \in \Delta (\varphi_i^\dag) \]

Let $\vec{a} = (\pi_{j = 1}^{i - 1})^\dag \in \prod_{j = 1}^{i - 1}
X_j$. Since $\pi$ is an optimal strategy for $\mathcal{G}$ we have
\[ q(\vec{a}, b^{\vec{a}}_i, \ldots, b^{\vec{a}}_n) \in
\varphi_i (\lambda x^{X_i} . q(\vec{a}, x, b^{\vec{a}, x}_{i + 1},
\ldots, b^{\vec{a}, x}_n)) \]

By induction on $j$ we have
\[ b^{\vec{a}}_j = \pi_j (\vec{a}, b^{\vec{a}}_i, \ldots, b^{\vec{a}}_{j -
  1}) = \pi_j (\pi_1^\dag, \ldots, \pi_{i - 1}^\dag, \pi_i^\dag,
\ldots, \pi_{j - 1}^\dag) = \pi_j^\dag \]
therefore
\[ q^\dag (\pi) = q (\pi^\dag) = q(\vec{a}, b^{\vec{a}}_i, \ldots,
b^{\vec{a}}_n) \]

We also have
\begin{align*}
&\varphi_i^\dag (\U_{q^\dag}^i (\pi)) \\
=\ &\varphi_i (\lambda x^{X_i} . \U_{q^\dag}^i (\pi) (\lambda \sigma^{P_i}
. x)) &&(\hbox{definition of } \varphi_i^\dag) \\
=\ &\varphi_i (\lambda x^{X_i} . q^\dag (\pi (i \mapsto \lambda
\sigma^{P_i} . x))) &&(\hbox{definition of } \U_{q^\dag}^i) \\
=\ &\varphi_i (\lambda x^{X_i} . q (\tau^\dag)) &&(\hbox{definition of
} q^\dag)
\end{align*}
where $\tau = \pi (i \mapsto \lambda \sigma^{P_i}
. x)$. By induction on $j$ we have
\[ \tau_j^\dag = \begin{cases}
\pi_j^\dag &\hbox{ if } 1 \leq j < i \\
x &\hbox{ if } i = j \\
\pi_j (\pi_1^\dag, \ldots, \pi_{i - 1}^\dag, x, \tau_{i + 1}^\dag,
\ldots, \tau_{j - 1}^\dag) &\hbox{ if } i < j \leq n
\end{cases} \]
We certainly have that $\tau^\dag$ coincides with $(\vec{a}, x,
b^{\vec{a}, x}_{i + 1}, \ldots, b^{\vec{a}, x}_n)$ at indices $1 \leq
j \leq i$. Moreover by induction on $i < j \leq n$ we have
\[ b^{\vec{a}, x}_j = \pi_j (\vec{a}, x, b^{\vec{a}, x}_{i + 1},
\ldots, b^{\vec{a}, x}_{j - 1}) = \pi_j (\vec{a}, x, \tau_{i +
  1}^\dag, \ldots, \tau_{j - 1}^\dag) = \tau_j^\dag \]
therefore
\[ \tau^\dag = (\vec{a}, x, b^{\vec{a}, x}_{i + 1}, \ldots,
b^{\vec{a}, x}_n) \]

We have therefore proven
\[ q^\dag (\pi) \in \varphi_i^\dag (\U_{q^\dag}^i (\pi)) \]
that is,
\[ (\U_{q^\dag}^i (\pi), \pi_i) \in \Delta(\varphi_i^\dag) \qedhere \]
\end{proof}

The converse is false because optimal strategies of
sequential games generalise the classical notion of subgame-perfect
equilibrium, which is a stronger condition than classical Nash
equilibrium (called an \emph{equilibrium refinement} in classical game
theory) \citep{escardo12}.

\section{2-player games and minimax strategies}

In this section the abstract notion of a $\psi$-$\varphi$ strategy is
defined, and used to show an intriguing connection between generalised
simultaneous games and proof theory. The reason for this terminology
is that a $\psi$-$\varphi$ strategy corresponds to a minimax strategy
in a 2-player game
with $\varphi = \max$ and $\psi = \min$. Note however that when
modelling a classical game as a generalised game all the quantifiers
will be $\max$, and so $\psi$-$\varphi$ strategies are distinct in
this sense from minimax strategies.

\begin{definition}[$\psi$-$\varphi$ strategy]
Let $\mathcal G$ be a $2$-player game with quantifiers $\varphi \in
S_{R_X} X$ and $\psi \in S_{R_Y} Y$. A
strategy $a \in X$ is called a \emph{$\psi$-$\varphi$ strategy} for player $1$
iff
\[ q_1(a, f(a)) \in \varphi (\lambda x^X . q_1(x, f(x))) \]
for all $f \in X \to Y$ with the property that for all $x \in X$,
\[ q_2(x, f(x)) \in \psi (\lambda y^Y . q_2(x, y)) \]
Similarly $b$ is a $\psi$-$\varphi$ strategy for player $2$ iff $q_2(g(b), b)
\in \psi (\lambda y^Y . q_2(g(y), y))$ whenever $q_1(g(y), y) \in \varphi
(\lambda x^X . q_1(x, y))$.
A \emph{$\psi$-$\varphi$ strategy profile} is one whose components are
both $\psi$-$\varphi$.
\end{definition}

The \emph{binary Berardi-Bezem-Coquand functional} is the higher-order function
\[ \bbc \in J_R X \times J_R Y \to J_R (X \times Y) \]
given by
\[ (\varepsilon \bbc \delta) (q) = (a, b) \]
where
\begin{align*}
a &= \varepsilon (\lambda x^X . q(x, \delta (\lambda y^Y . q(x, y))))
\\
b &= \delta (\lambda y^Y . q(\varepsilon (\lambda x^X . q(x, y)), y))
\end{align*}
Notice that the type of $\bbc$ is the same as the type of
$\otimes$. Moreover when infinitely iterated both provide proof interpretations (in the
modified realizability interpretation of Heyting arithmetic) of the
axiom of countable choice \citep{berardi98, berger02} and a certain
equivalence, namely interdefinability over system T, is shown in
\citep{powell12}. However the relationship
between the two functionals is not well understood, and only the
product of selection functions has previously been linked to game theory.

\begin{theorem}
Let $\mathcal{G}$ be a $2$-player game with single outcome space such
and total single-valued quantifiers attained by
selection functions $\varepsilon$, $\delta$. Then the strategy profile
\[ (\varepsilon \bbc \delta) (q) \in X \times Y \]
is a $\psi$-$\varphi$ strategy profile.
\end{theorem}

\begin{proof}
Since $\psi$ is single-valued, the unique $f$ with the
given property is
\[ f(x) = b_x = \delta (\lambda y^Y . q(x, y)) \]
The first component in the strategy profile is
\[ a = \varepsilon (\lambda x^X . q(x, b_x)) \]
We therefore have
\[ q(a, b_a) = (\lambda x^X . q(x, b_x)) (\varepsilon (\lambda x^X
. q(x, b_x))) \in \varphi (\lambda x^X . q(x, b_x)) \]
as required. The proof for $b$ is symmetric.
\end{proof}

In particular, in a 2-player classical game if a player has outcome
function $q$ then they have a minimax strategy given by $(\argmax \bbc
\argmin)(q)$. 

\section{Conclusions}

For sequential games, the proof of the existence of equilibria uses
the product of selection functions and is constructive (theorem
\ref{optimal-strategies}). Due to the importance of the product of
selection functions, the constructive nature of the existence proof is
an important part of the theory. Theorem \ref{generalised-nash-theorem} is
similar to theorem \ref{optimal-strategies} but is nonconstructive.

Nash gave 2
different proofs of his existence theorem, one using the
Brouwer fixed point theorem \citep{nash50a} and the other using
the Kakutani fixed point theorem, a weaker form of the
Kakutani-Fan-Glicksberg theorem applicable only to Euclidean spaces
\citep{nash50b}. Of these, only the second appears to be amenable to
generalising as we have done. The Brouwer, Kakutani and Nash theorems
are all known to not be provable constructively but they
all have equivalent approximation theorems that are provable
constructively and are complete for the same complexity class
\citep{tanaka11, daskalakis06}.
The author plans to investigate the computation of abstract
Nash equilibria in subsequent papers.

It should be noted that theorem \ref{generalised-nash-theorem}, like
Nash's original theorem, has a simpler proof using the weaker Kakutani
theorem. The reason for proving the stronger theorem
\ref{multilinear-existence-thm}
is that it can be used to prove a stronger result which also
generalises Glicksberg's theorem \citep{glicksberg52}, a result which
generalises Nash's theorem for finite games to games whose strategy
sets are compact topological spaces and whose outcome functions are
continuous.

\bibliographystyle{plainnat}
\bibliography{/Users/juleshedges/Documents/Work/refs}

\end{document}